\documentclass[11pt]{amsart}
\usepackage{amssymb,latexsym}
\topskip  \newtheorem{theorem}{Theorem}

\newtheorem{corollary}[theorem]{Corollary}

\newtheorem{remark}[theorem]{Remark}
\newtheorem{lemma}[theorem]{Lemma}

\newtheorem{example}[theorem]{Example}

\begin{document}

\pagenumbering{arabic}
\pagestyle{headings}
\def\sof{\hfill\rule{2mm}{2mm}}
\def\ls{\leq}
\def\gs{\geq}
\def\sl{\sum\limits}
\def\SS{S}
\def\SSS{S\hspace{-2pt}h}
\def\qq{{\bold q}}
\def\txx{{\frac1{2\sqrt{x}}}}
\def\AA{\mathcal{A}}
\def\aa{({\mathbf\alpha})}

\title{Some recursive formulas for Selberg-type integrals}
\author{Sergio Iguri}
\address{Instituto de Astronom\'{\i}a y F\'{\i}sica del Espacio (CONICET-UBA). \\ C.~C.~67 - Suc.~28, 1428 Buenos Aires, Argentina.}
\email{siguri@iafe.uba.ar}
\author{Toufik Mansour}
\address{Department of Mathematics, University of Haifa, Haifa 31905, Israel}
\address{Mathematical Science, G\"oteborg University and Chalmers University of Technology, S-412 96 G\"oteborg, Sweden}
\email{toufik@math.haifa.ac.il}

\begin{abstract}
A set of recursive relations satisfied by Selberg-type integrals involving monomial symmetric polynomials are derived, generalizing previous results in \cite{Ao,Ig}. These formulas provide a well-defined algorithm for computing Selberg-Schur integrals whenever the Kostka numbers relating Schur functions and the corresponding monomial polynomials are explicitly known. We illustrate the usefulness of our results discussing some interesting examples.
\end{abstract}

\maketitle
%
\medskip

{\sc 2000 Mathematics Subject Classification}: 33D70, 81T30, 81T40, 81V70

\section{Introduction}
The Selberg integral and its generalizations have played a central role both in pure and applied mathematics. Their applications run from the proof of the Mehta-Dyson conjecture and several cases of the Macdonald conjectures \cite{A,M,Mor} to the study of some $q$-analogues of constant term identities, through Calogero-Sutherland quantum many body models \cite{E,K,K2,Ka,S,W}, orthogonal polynomials theory \cite{O,St}, hyperplane arrangements \cite{Sch} and random matrix theory \cite{DT,F,KS,KLR}. They also have a deep connection to the Knizhnik–-Zamolodchikov equations \cite{MV,W2} with the corresponding implications in conformal field theory and even string theory \cite{DV,EFK,FSV,Ig2,Ig3,MT,TV,TK}. See \cite{F2} for a comprehensive review on the relevance of the Selberg integral and its applications.

The aim of this paper is to study the Selberg-type integral with the integrand dressed up with a symmetric function, namely, to study integrals of the form $J_{f}\equiv J^{(N)}(a,b,\rho;f)$:
\begin{equation}\label{eq1}
J_{f}=\int_{\Lambda}f(y_1,\ldots,y_N)\prod_{i=1}^Ny_i^{a-1}(1-y_i)^{b-1}\prod_{1\leq i<j\leq N}|y_i-y_j|^{2\rho}dy_1\wedge\cdots\wedge dy_N,
\end{equation}
where $f(y_1,\ldots,y_N)$ is a symmetric polynomial, the integral is taken over the $N$-dimensional open domain\footnote{Selberg-type integrals are sometimes defined over the $N$-dimensional simplex $\{(y_1,\ldots,y_N)\in\mathbb{R}^N\mid 0<y_1<\cdots<y_N<1\}$. From the symmetry of the integrand under a permutation of the variables we get that these integrals differ by a factor of $1/N!$.} $\Lambda=(0,1)^N$ and $a,b$ and $\rho$ are complex numbers. For simplicity, we will denote the function
$\prod_{i=1}^Ny_i^{a-1}(1-y_i)^{b-1}\prod_{1\leq i<j\leq N}|y_i-y_j|^{2\rho}$ by $\Phi(y)$, the $N$-form $dy_1\wedge\cdots\wedge dy_N$ by $dy$ and the polynomial $f(y_1,\ldots,y_N)$ by $f(y)$.

Among the basis for the space of symmetric polynomials, Schur basis plays a special role in this context. The importance of Selberg-Schur integrals was stated in \cite{TK} when studying the non-triviality of the integral representation of the intertwining operators between the Fock space representations of the Virasoro algebra and in \cite{B}, in a more general setting, when analyzing the Fock space resolutions of the $\widehat{sl}(n)$ irreducible highest-weight modules. As expected, they also appear when computing correlation functions on the sphere in related Wess-Zumino-Novikov-Witten models \cite{Ig2,Ig3}. Given a partition $\lambda$ we will denote  the Schur polynomial associated with it by $s_{\lambda}(y)$ and the corresponding Selberg-Schur integral by $J_{\lambda}$.

The case $\lambda=0$ corresponds to the classical integral considered by Selberg in \cite{Sel}. The analytic expression he found for this integral is:
\begin{eqnarray}
J_{0}=\int_\Lambda\Phi(y)dy=\prod_{i=1}^N\frac{\Gamma(a+(N-i)\rho)\Gamma(b+(N-i)\rho)\Gamma(i\rho+1)}{\Gamma(a+b+(2N-i-1)\rho)\Gamma(\rho+1)},
\end{eqnarray}
and it is well defined whenever $a$, $b$ and $\rho$ satisfy
\begin{eqnarray}
\Re(a),\Re(b)>0 \qquad \mbox{and} \qquad \Re(\rho)>-\min\left\{\frac{1}{N},\frac{\Re(a)}{N-1},\frac{\Re(b)}{N-1}\right\},
\end{eqnarray}
the second inequality having meaning for $N>1$. From now on we assume that these conditions always hold.

When $\lambda=(1^{m_1})$ with $0\le m_1\leq N$, Schur polynomials reduce to elementary symmetric polynomials, {\em i.e.},
\begin{eqnarray}
s_{(1^{m_1})}(y)\equiv e_{m_1}(y)=\frac{1}{N!}\binom{N}{{m_1}}\sum_{\sigma\in\SS_N}\prod_{i=1}^{m_1} y_{\sigma(i)},
\end{eqnarray}
where $\SS_N$ is the set of permutations of the set $\{1,2,\ldots,N\}$ and $e_0(y)=1$. In this case, Aomoto \cite{Ao} showed that
\begin{equation}
\label{eq2}
J_{(1^{m_1})}=\int_{\Lambda}e_{m_1}(y)\Phi(y)dy=J_0\binom{N}{{m_1}}\prod_{i=1}^{m_1}\frac{a+(N-i)\rho}{a+b+(2N-1-i)\rho}.
\end{equation}

A further extension of Selberg integral, by far the most general one, has been computed by Kadell in \cite{K3} and it involves Jack functions. It reads:
\begin{equation}
\label{eq3}
\int_{\Lambda}P_\lambda^{(1/\rho)}(y)\Phi(y)dy=J_0P_\lambda^{(1/\rho)}(1^N)\frac{[a+(N-1)\rho]_\lambda^{(\rho)}}{[a+b+2(N-1)\rho]_\lambda^{(\rho)}},
\end{equation}
where $\lambda$ is an arbitrary partition, $P_\lambda^{(1/\rho)}(y)$ is a Jack polynomial and $[a]_\lambda^{(\rho)}$ is a generalized Pochhammer symbol, which is defined as
\begin{eqnarray}
[a]_\lambda^{(\rho)}=\prod_{i\geq1}(a+(1-i)\rho)_{\lambda_i},
\end{eqnarray}
$(a)_n$ being the standard Pochhammer symbol, namely, $(a)_n=a(a+1)\cdots(a+n-1)$ with $(a)_0=1$. When $\rho=1$ we have $P_\lambda^{(1/\rho)}(y)=s_\lambda(y)$ so that (\ref{eq3}) gives
\begin{equation}\label{eq4}
J_{\lambda}=J_0s_\lambda(1^N)\frac{[a+(N-1)]_\lambda^{(1)}}{[a+b+2N-2]_\lambda^{(1)}}.
\end{equation}

More recently, it was proved in \cite{Ig} that for the case $\lambda=(2^{m_2}1^{m_1})$, $0 \le m_1 + m_2 \leq N$, one has
\begin{eqnarray}\label{eq5}
J_{(2^{m_2}1^{m_1})}=J_0m_\lambda(1^N)\dfrac{[a+(N-1)\rho]_\lambda^{(\rho)}}{[a+b+2(N-1)\rho]_\lambda^{(\rho)}}
\dfrac{[a+b+(N-2)\rho]_{(1^{m_2})}^{(\rho)}}{[a+b+(2N-m_1-m_2-2)\rho]_{(1^{m_2})}^{(\rho)}}\\
\qquad\qquad\times_4F_3\left[\begin{array}{c}-m_2, -N+m_1+m_2,\alpha+\beta+\gamma+2N-m_2-1,\alpha+N-m_2+1\\ \alpha+\beta+N-m_2-1,\alpha+\gamma+N-m_2,m_1+2\end{array}\right],\nonumber
\end{eqnarray}
where $\alpha=a/\rho$, $\beta=b/\rho$, $\gamma=1/\rho$, the hypergeometric series $_4F_3$ is evaluated at $1$ and $m_\lambda(y)$ denotes the monomial symmetric polynomial associated to the partition $\lambda$.

In this paper, using similar techniques as those employed in \cite{Ao,Ig}, we find a set of recursive formulas satisfied by generic Selberg-type integrals involving monomial polynomials. These recursions and the fact that Schur polynomials can be uniquely decomposed as linear combinations of monomial symmetric functions reduce the problem of computing (\ref{eq1}) to the problem of computing Kostka numbers while providing a well-defined algorithm for obtaining Selberg-Schur integrals in the general case.

The paper is organized as follows. After introducing some notation we prove in Section 2 some lemmas and preliminary propositions that will be useful for obtaining in Section 3 the recursive relations we have already announced. In Section 4 we illustrate the usefulness of our results with several relevant examples.

\section{Notation and preliminary lemmas}

In this section we fix our conventions, we introduce some notation and we derive several formulas that will be needed in order to prove our main results.

Despite of the fact that partitions are usually defined without trivial components, it will be useful for our purposes to identify partitions with length $\ell_{\lambda}\le N$ with decreasingly ordered $N$-tuples with non-negative entries by defining $\lambda_i=0$ for $i=\ell_{\lambda}+1,\dots,N$.

Given $v\in \mathbb{R}^N$, $v=(v_1,\ldots,v_N)$, we define its (decreasingly) ordered partner $[v]$ as the vector $(v_{\sigma(1)},\ldots,v_{\sigma(N)})$, where $\sigma \in \SS_N$ is any permutation satisfying $v_{\sigma(1)}\geq v_{\sigma(2)}\geq\cdots\geq v_{\sigma(N)}$. If $v_1,\ldots,v_N$ are all non-negative integer numbers, then $[v]$ actually defines a partition with length $\ell_{[v]}\le N$. We denote the standard basis of $\mathbb{R}^N$ by $\{e_1,\ldots,e_N\}$, $e_j$ being the $j$-th unit vector.

For a partition $\lambda$ let us denote by $\lambda'$ its conjugate so that $\lambda'_k$ gives the number of entries $\ge k$ in $\lambda$. Notice that $\ell_{\lambda}=\lambda'_1$. If $y=(y_1,\dots,y_N)$ we define
\begin{eqnarray}
y^{\lambda} = \prod_{j=1}^{N} y_j^{\lambda_j} = \prod_{i=0}^n \prod_{r=1}^{m_{n-i}} y^{n-i}_{\lambda'_{n-i+1}+r},
\end{eqnarray}
where $n$ is the greatest part of $\lambda$, $m_{k}=\lambda'_{k}-\lambda'_{k+1}$, $k=1,\dots,n$, is the multiplicity of the part $k$ in $\lambda$ and $m_0=N$. Further, let us define the following integrals:
\begin{equation}\label{eqa1}
B_{\lambda}=\int_{\Lambda}y^{\lambda}\Phi(y)dy,
\end{equation}
and, for any integer number $c\ge0$,
\begin{equation}\label{eqa2}
A_{\lambda}(k,c)=\int_{\Lambda}\frac{y_1^c \prod_{j=2}^N y_j^{\lambda_j}}{y_1-y_k}\Phi(y)dy,
\end{equation}
and
\begin{equation}
\label{eqa2b}
K_{\lambda}(c)=\int_{\Lambda}\frac{y_1^c \prod_{j=2}^N y_j^{\lambda_j}}{1-y_1}\Phi(y)dy.
\end{equation}
We will denote $A_{\lambda}(k,\lambda_1)$ simply by $A_{\lambda}(k)$.

We will generalize \cite[Lemma 1]{Ig} and \cite[Lemma 3]{Ig} by proving the following.

\begin{lemma}\label{lema3}
Let $\lambda$ be a partition such that $\ell_{\lambda}\leq N$ and let $c$ be a non-negative integer number. Let $2 \le k \le N$. Then,
\begin{eqnarray}
A_{\lambda}(k,c)=\left\{
\begin{array}{lll}
-\frac{1}{2}\sum\limits_{i=0}^{\lambda_k-1-c}B_{[\lambda+(c+i-\lambda_1)e_1-(1+i)e_k]}& \mbox{if}& c < \lambda_k, \\
0&\mbox{if}&c=\lambda_k, \\
\frac{1}{2}\sum\limits_{i=0}^{c-\lambda_k-1}B_{[\lambda+(c-1-i-\lambda_1)e_1+ie_k]}&\mbox{if}&c > \lambda_k.
\end{array}\right.
\end{eqnarray}
\end{lemma}
\begin{proof}
Exchanging $y_k$ and $y_1$ in (\ref{eqa2}) and then using the symmetry of Selberg's kernel $\Phi(y)$ under the permutation of any pair of variables, we obtain
\begin{eqnarray}
A_{\lambda}(k,c)=-\int_{\Lambda}\frac{y_1^{\lambda_k}y_k^c\prod_{j\neq 1,k}^{N} y_j^{\lambda_j}}{y_1-y_k}\Phi(y)dy.
\end{eqnarray}

Thus, if $0\leq c < \lambda_k$, we get
\begin{eqnarray}
A_{\lambda}(k,c)=-\frac{1}{2}\int_{\Lambda}\frac{y_1^cy_k^c(y_1^{\lambda_k-c}-y_k^{\lambda_k-c})\prod_{j\neq 1,k}^{N} y_j^{\lambda_j}}{y_1-y_k}\Phi(y)dy,
\end{eqnarray}
which is equivalent to
\begin{eqnarray}
A_{\lambda}(k,c)=-\frac{1}{2}\sum_{i=0}^{\lambda_k-c-1}B_{[\lambda+(c+i-\lambda_1)e_1-(1+i)e_k]},
\end{eqnarray}
where we have used
\begin{eqnarray}
y_1^{\lambda_k-c}-y_k^{\lambda_k-c} = (y_1-y_k) \sum_{i=0}^{\lambda_k-c-1} y_1^{\lambda_k-c-1-i}y_k^{i}.
\end{eqnarray}

When $c=\lambda_k$ it is straightforward to see that integral (\ref{eqa2}) vanishes.

If, instead, $c>\lambda_k$, then
\begin{eqnarray}
A_{\lambda}(k,c)=\frac{1}{2}\int_{\Lambda}\frac{y_1^{\lambda_k}y_k^{\lambda_k}(y_1^{c-\lambda_k}-y_k^{c-\lambda_k})\prod_{j\neq 1,k}^{N} y_j^{\lambda_j}}{y_1-y_k}\Phi(y)dy,
\end{eqnarray}
namely,
\begin{eqnarray}
A_{\lambda}(k,c)=\frac{1}{2}\sum_{i=0}^{c-\lambda_k-1}B_{[\lambda+(c-1-i-\lambda_1)e_1+ie_k]},
\end{eqnarray}
as we wanted to prove.
\end{proof}

\begin{corollary}
\label{coro1}
Let $\lambda$ be a partition with $\ell_{\lambda}\leq N$ and let $2 \le k \le N$. Then,
\begin{eqnarray}
A_{\lambda}(k)=\frac{1}{2}\sum_{i=0}^{\lambda_1-\lambda_k-1}B_{[\lambda-(1+i)e_1+ie_k]}.
\end{eqnarray}
\end{corollary}

\begin{example}
\label{ex3}
If $\lambda=(2^{m_2}1^{m_1})$, $m_1+m_2\leq N$, then
\begin{equation}\label{eq21a3bis}
A_{(2^{m_2}1^{m_1})}(k)=\left\{
\begin{array}{lll}
0&\mbox{if}&2\leq k\leq \lambda'_2,\\
\frac{1}{2}B_{(2^{m_2-1}1^{m_1+1})}&\mbox{if}&\lambda'_2+1\leq k\leq \lambda'_1,\\
B_{(2^{m_2-1}1^{m_1+1})}&\mbox{if}&\lambda'_1+1\leq k\leq N,
\end{array}\right.
\end{equation}
while Lemma~\ref{lema3} gives, for $c=1$,
\begin{equation}\label{eq21a3}
A_{(2^{m_2}1^{m_1})}(k,1)=\left\{
\begin{array}{lll}
-\frac{1}{2}B_{(2^{m_2-2}1^{m_1+2})}&\mbox{if}&2\leq k\leq \lambda'_2,\\
0&\mbox{if}&\lambda'_2+1\leq k\leq \lambda'_1,\\
\frac{1}{2}B_{(2^{m_2-1}1^{m_1-1})}&\mbox{if}&\lambda'_1+1\leq k\leq N.
\end{array}\right.
\end{equation}
as it was already shown in \cite[Lemma 1]{Ig} and \cite[Lemma 3]{Ig}, respectively. Furthermore, we find
\begin{equation}\label{eq21a3t}
A_{(2^{m_2}1^{m_1})}(k,0)=\left\{
\begin{array}{lll}
-B_{(2^{m_2-2}1^{m_1+1})}&\mbox{if}&2\leq k\leq \lambda'_2,\\
-\frac{1}{2}B_{(2^{m_2-1}1^{m_1-1})}&\mbox{if}&\lambda'_2+1\leq k\leq \lambda'_1,\\
0&\mbox{if}&\lambda'_1+1\leq k\leq N.
\end{array}\right.
\end{equation}
\end{example}

\begin{example}
If $\lambda=(3^{m_3} 2^{m_2} 1^{m_1})$, $\ell_{\lambda} \leq N$, then
\begin{equation}\label{eq321a3c1bis}
A_{(3^{m_3} 2^{m_2} 1^{m_1})}(k)=\left\{
\begin{array}{lll}
0&\mbox{if}&2\leq k\leq \lambda'_3,\\
\frac{1}{2}B_{(3^{m_3-1}2^{m_2+1}1^{m_1})}&\mbox{if}&\lambda'_3+1\leq k\leq \lambda'_2,\\
B_{(3^{m_3-1}2^{m_2+1}1^{m_1})}&\mbox{if}&\lambda'_2+1\leq k\leq \lambda'_1,\\
B_{(3^{m_3-1}2^{m_2+1}1^{m_1})}\\ \qquad \quad +\frac{1}{2}B_{(3^{m_3-1}2^{m_2}1^{m_1+2})}&\mbox{if}&\lambda'_1+1\leq k\leq N,
\end{array}\right.
\end{equation}
and Lemma~\ref{lema3} gives
\begin{equation}
\label{eq321a3c2}
A_{(3^{m_3} 2^{m_2} 1^{m_1})}(k,2)=\left\{
\begin{array}{lll}
-\frac{1}{2}B_{(3^{m_3-2}2^{m_2+2}1^{m_1})}&\mbox{if}&2\leq k\leq \lambda'_3,\\
0&\mbox{if}&\lambda'_3+1\leq k\leq \lambda'_2,\\
\frac{1}{2}B_{(3^{m_3-1}2^{m_2}1^{m_1+1})}&\mbox{if}&\lambda'_2+1\leq k\leq \lambda'_1,\\
B_{(3^{m_3-1}2^{m_2}1^{m_1+1})}&\mbox{if}&\lambda'_1+1\leq k\leq N,
\end{array}\right.
\end{equation}
\begin{equation}
\label{eq321a3c1}
A_{(3^{m_3} 2^{m_2} 1^{m_1})}(k,1)=\left\{
\begin{array}{lll}
-B_{(3^{m_3-2}2^{m_2+1}1^{m_1+1})}&\mbox{if}&2\leq k\leq \lambda'_3,\\
-\frac{1}{2}B_{(3^{m_3-1}2^{m_2-1}1^{m_1+2})}&\mbox{if}&\lambda'_3+1\leq k\leq \lambda'_2,\\
0&\mbox{if}&\lambda'_2+1\leq k\leq \lambda'_1,\\
\frac{1}{2}B_{(3^{m_3-1}2^{m_2}1^{m_1})}&\mbox{if}&\lambda'_1+1\leq k\leq N,
\end{array}\right.
\end{equation}
and
\begin{equation}
\label{eq321a3c1ot}
A_{(3^{m_3} 2^{m_2} 1^{m_1})}(k,0)=\left\{
\begin{array}{lll}
-B_{(3^{m_3-2}2^{m_2+1}1^{m_1})}\\ \qquad \quad-\frac{1}{2}B_{(3^{m_3-2}2^{m_2}1^{m_1+2})}&\mbox{if}&2\leq k\leq \lambda'_3,\\
-B_{(3^{m_3-1}2^{m_2-1}1^{m_1+1})}&\mbox{if}&\lambda'_3+1\leq k\leq \lambda'_2,\\
-\frac{1}{2}B_{(3^{m_3-1}2^{m_2}1^{m_1-1})}&\mbox{if}&\lambda'_2+1\leq k\leq \lambda'_1,\\
0&\mbox{if}&\lambda'_1+1\leq k\leq N.
\end{array}\right.
\end{equation}
\end{example}

Concerning integrals (\ref{eqa2b}) we can prove the following two lemmas.

\begin{lemma}
\label{lema2}
Let $\lambda$ be any partition with $\ell_{\lambda}\leq N$. Thus, for $0\leq c\leq\lambda_1$,
\begin{equation}
K_\lambda(c)=K_\lambda(0)-\sum_{i=0}^{c-1}B_{[\lambda+(i-\lambda_1)e_1]}.
\end{equation}
\end{lemma}
\begin{proof}
The proof of the lemma follows straightforwardly after using the substitution
\begin{eqnarray}
\frac{y_1^c}{1-y_1}=\frac{1}{1-y_1}-\sum_{i=0}^{c-1}y_1^i
\end{eqnarray}
in Eq.~(\ref{eqa2b}).
\end{proof}

\begin{lemma}\label{lema4}
Let $\lambda$ be any partition with $\ell_{\lambda}\leq N$. For $0\leq c\leq\lambda_1$ we have
\begin{eqnarray}\label{ll4}
K_\lambda(c)=\frac{2\rho}{b-1}\sum_{k=2}^N A_{\lambda}(k,c)+\frac{a-1+c}{b-1}B_{[\lambda+(c-1-\lambda_1)e_1]}.
\end{eqnarray}
\end{lemma}
\begin{proof}
Since $\Phi(y)$ vanishes at the boundary values $y_1=0$ and $y_1=1$, we obtain after applying Stokes' theorem,
\begin{eqnarray}
0&=&\int_{\Lambda}d_1\left(y_1^c\prod_{j=2}^\ell y_j^{\lambda_j}\Phi(y)dy'\right) \\
&=& 2\rho\sum_{k=2}^N A_{\lambda}(k,c)+(a-1+c)B_{[\lambda+(c-1-\lambda_1)e_1]}-(b-1)K_\lambda(c),\nonumber\label{l4}
\end{eqnarray}
which follows from the fact that
\begin{eqnarray}
d_1 \Phi(y) = \frac{a}{y_1}-\frac{b-1}{1-y_1}+2\rho\sum_{k=2}^N\frac{1}{y_1-y_k}.
\end{eqnarray}
Eq.~(\ref{ll4}) follows from (\ref{l4}).
\end{proof}

The following example essentially reproduces the derivation of the recurrence found in \cite{Ig} for Selberg integrals involving symmetric monomial polynomials associated to partitions with entries $\le 2$, namely, \cite[Lemma 4]{Ig}.

\begin{example}\label{ex7}
Let $\lambda=(2^{m_2}1^{m_1})$ and $c=2$. Lemma~\ref{lema3} gives
\begin{eqnarray}
(b-1)K_{(2^{m_2}1^{m_1})}(2)=2\rho\sum_{k=2}^N A_{\lambda}(k,2)+(a+1)B_{(2^{m_2-1}1^{m_1+1})}.
\end{eqnarray}
By virtue of Example \ref{ex3} we get
\begin{equation}
(b-1)K_{(2^{m_2}1^{m_1})}(2)=(a+1+\rho(2N-m_1-2m_2))B_{(2^{m_2-1}1^{m_1+1})}.
\end{equation}
Using Lemma~\ref{lema2} we find
\begin{eqnarray}
(b-1)\left(K_{(2^{m_2}1^{m_1})}(0)-B_{(2^{m_2-1}1^{m_1})}-B_{(2^{m_2-1}1^{m_1+1})}\right) \\
= (a+1+\rho(2N-2m_2-m_1))B_{(2^{m_2-1}1^{m_1+1})},\nonumber
\end{eqnarray}
which is equivalent to
\begin{eqnarray}\label{eqex21a1}
(b-1)K_{(2^{m_2}1^{m_1}}(0) = (a+b+\rho(2N-m_1-2m_2))B_{(2^{m_2-1}1^{m_1+1})} \\
 +(b-1)B_{(2^{m_2-1}1^{m_1})},\nonumber
\end{eqnarray}
as proved in \cite[lemma 2]{Ig}.

In a similar way, for $c=1$, Lemma~\ref{lema3} and Lemma~\ref{lema2} give
\begin{eqnarray}
\qquad (b-1)\left(K_{(2^{m_2}1^{m_1})}(0)-B_{(2^{m_2-1}1^{m_1})}\right) = 2\rho\sum_{k=2}^N A_{\lambda}(k,1)+a B_{(2^{m_2-1}1^{m_1})},
\end{eqnarray}
which is equivalent to
\begin{eqnarray}\label{eqex21a2}
(b-1)K_{(2^{m_2}1^{m_1})}(0)=-\rho(m_2-1)B_{(2^{m_2-2}1^{m_1+2})}\\
+(a+b-1+\rho(N-m_1-m_2))B_{(2^{m_2-1}1^{m_1})}.\nonumber
\end{eqnarray}

After combining (\ref{eqex21a1}) and (\ref{eqex21a2}) we obtain
\begin{eqnarray}\label{formula}
(a+b+\rho(2N-2m_2-m_1))B_{(2^{m_2-1}1^{m_1+1})}=(a+\rho(N-m_1-m_2))B_{(2^{m_2-1}1^{m_1})}\\
-\rho(m_2-1)B_{(2^{m_2-2}1^{m_1+2})},\nonumber
\end{eqnarray}
as it is proved in \cite[Lemma 4, Equation (13)]{Ig}.
\end{example}

\section{Recurrence relations for Selberg-type integrals}

Formula (\ref{formula}) defines a recursive relation that was used in \cite{Ig} for computing Selberg-Schur integrals associated to partitions of the form $\lambda=(2^{m_2}1^{m_1})$, $0 \le m_1+m_2 \le N$. In this section we find a set of recurrence relations satisfied by Selberg integrals involving monomial polynomials generalizing \cite[Lemma 4, Equation (13)]{Ig}.

\begin{theorem}\label{thm}
Let $\lambda$ be any partition of length $\ell_{\lambda}<N$. Then, for any $c$ such that $0 \leq c < \lambda_1$ we have
\begin{eqnarray}\label{eqB}
(b-1)\sum\limits_{i=c}^{\lambda_1-1}B_{[\lambda+(i-\lambda_1)e_1]}+(a-1+\lambda_1)B_{[\lambda-e_1]}
-(a-1+c)B_{[\lambda+(c-1-\lambda_1)e_1]} \\
= \rho\sum_{k=2}^N (-1)^{\delta_{\lambda_k < c}} \sum_{i=1}^{\max\{\lambda_k,c\}-\min\{\lambda_k,c\}} B_{[\lambda+(\max\{\lambda_k,c\}-i-\lambda_1)e_1+(\min\{\lambda_k,c\}+i-1-\lambda_k)e_k]}\nonumber \\
-\rho\sum_{k=2}^N \sum_{i=1}^{\lambda_1-\lambda_k} B_{[\lambda-ie_1+(i-1)e_k]},\nonumber
\end{eqnarray}
where $\delta_{a<b}$ equals $0$ if $a<b$ and it equals $1$ otherwise.
\end{theorem}

\begin{remark}
Before proving the Theorem let us emphasize that (\ref{eqB}) is actually a well defined recurrence for $c < \lambda_1$ for the dominance ordering on partitions, namely, all partitions appearing in (\ref{eqB}) are $\preceq \lambda$.
\end{remark}

\begin{proof}
The proof of the theorem follows the same steps as Example \ref{ex7}. Using Lemma~\ref{lema2} for $c=\lambda_1$ we obtain
\begin{eqnarray}
\qquad (b-1)\left(K_\lambda(0)-\sum_{i=0}^{\lambda_1-1}B_{[\lambda+(i-\lambda_1)e_1]}\right) = 2\rho\sum_{k=2}^N A_{\lambda}(k)+(a-1+\lambda_1)B_{[\lambda-e_1]},
\end{eqnarray}
which is equivalent to
\begin{eqnarray}\label{eqth1}
(b-1)K_\lambda(0)=(b-1)\sum\limits_{i=0}^{\lambda_1-1}B_{[\lambda+(i-\lambda_1)e_1]}+(a-1+\lambda_1)B_{[\lambda-e_1]}\\
+\rho\sum\limits_{k=2}^N\sum\limits_{i=1}^{\lambda_1-\lambda_k}B_{[\lambda-ie_1+(i-1)e_k]}.\nonumber
\end{eqnarray}

On the other hand, Lemma~\ref{lema2} for an arbitrary $0 \le c < \lambda_1$ gives
\begin{eqnarray}
\qquad (b-1)\left(K_\lambda(0)-\sum_{i=0}^{c-1}B_{[\lambda+(i-\lambda_1)e_1]}\right)
=2\rho\sum_{k=2}^N A_{\lambda}(k,c)+(a-1+c)B_{[\lambda+(c-1-\lambda_1)e_1]},
\end{eqnarray}
and by Lemma~\ref{lema3} it follows that
\begin{eqnarray}
\sum\limits_{k=2}^N A_{\lambda}(k,c)=-\frac{1}{2}\sum\limits_{k=2,c\leq\lambda_k}^N \sum\limits_{i=0}^{\lambda_k-1-c}B_{[\lambda+(c+i-\lambda_1)e_1-(1+i)e_k]}\\
+\frac{1}{2}\sum\limits_{k=2,c>\lambda_k}^N\sum\limits_{i=0}^{c-\lambda_k-1}B_{[\lambda+(c-1-i-\lambda_1)e_1+ie_k]},\nonumber
\end{eqnarray}
so that
\begin{eqnarray}\label{eqth2}
(b-1)K_\lambda(0)=(b-1)\sum\limits_{i=0}^{c-1}B_{[\lambda+(i-\lambda_1)e_1]}+(a-1+c)B_{[\lambda+(c-1-\lambda_1)e_1]}\\
-\rho\sum\limits_{k=2,c\leq\lambda_k}^N \sum\limits_{i=0}^{\lambda_k-1-c}B_{[\lambda+(c+i-\lambda_1)e_1-(1+i)e_k]}
+\rho\sum\limits_{k=2,c>\lambda_k}^N\sum\limits_{i=0}^{c-\lambda_k-1}B_{[\lambda+(c-1-i-\lambda_1)e_1+ie_k]}.\nonumber
\end{eqnarray}

After combining Equations (\ref{eqth1}) and (\ref{eqth2}) we get the desired result.
\end{proof}

\begin{corollary} Let $\lambda$ be a partition of length $\ell_{\lambda}<N$ with $\lambda_1>\lambda_k$ for $k=2,\dots,N$. Then,
\begin{eqnarray}\label{eqB2}
(a+b+\lambda_1-2)B_{[\lambda-e_1]}
-(a+\lambda_1-2)B_{[\lambda-2e_1]} = -\rho\sum_{k=2}^N \sum_{i=1}^{\lambda_1-\lambda_k} B_{[\lambda-ie_1+(i-1)e_k]} \\
+ \rho\sum_{k=2}^N \sum_{i=1}^{\lambda_1-\lambda_k-1} B_{[\lambda-(i+1)e_1+(i-1)e_k]}.\nonumber
\end{eqnarray}
\end{corollary}
\begin{proof}
The proof of the corollary follows straightforwardly after replacing $c=\lambda_1-1$ in Eq.~(\ref{eqB}).
\end{proof}

In the next section we give some examples showing the usefulness of these results.

\section{Applications}

Let us recall that any Schur polynomial can be expressed in term of monomial symmetric polynomials as
\begin{equation}
s_\lambda(y)=\sum_{\mu\preceq\lambda}K_{\lambda\mu} m_\mu(y),
\end{equation}
where $K_{\lambda\mu}$ is the Kostka number associated with $\lambda$ and $\mu$ and $m_\mu(y)$ is the monomial symmetric polynomial indexed by $\mu$.

From the symmetry of the Selberg-Schur kernel under the permutation of any pair of variables it follows that
\begin{eqnarray}\label{eqJ1}
J_{\lambda}=\sum_{\mu\preceq\lambda} m_\mu(1^N) K_{\lambda\mu}B_\mu,
\end{eqnarray}
where we have exploit that $m_\mu(y)$ is a symmetric polynomial, thus proving that Theorem~\ref{thm} provides us with a well-defined algorithm for computing Selberg-Schur integral as it was announced in the introduction.

Let us illustrate this fact with some interesting examples.

\subsection{Partitions of the form $(31^n)$:} For a given positive integer $n$ let us consider the partition $(31^n)$. The partitions $\mu$ satisfying that $\mu\preceq (31^n)$ are: $(1^{n+3})$, $(21^{n+1})$ and $(31^n)$. Recall from \cite{Ao} that
\begin{equation}\label{eqb111}
B_{(1^n)}=J_0 \frac{[a+(N-1)\rho]^{(\rho)}_{(1^n)}}{[a+b+2(N-1)\rho]^{(\rho)}_{(1^n)}},
\end{equation}
and from \cite{Ig} that
\begin{equation}\label{eqb2211}
\begin{array}{l}
B_{(2^n1^{m})}=J_0\dfrac{[a+(N-1)\rho]_{(2^n1^{m})}^{(\rho)}}{[a+b+2(N-1)\rho]_{(2^n1^{m})}^{(\rho)}}\dfrac{[a+b+(N-2)\rho]_{(1^n)}^{(\rho)}}{[a+b+(2N-m-n-2)\rho]_{(1^n)}^{\rho}}\\
\qquad\qquad\qquad\qquad\qquad\qquad\times _3F_2\left[\begin{array}{c}-n,-N+m+n,\alpha+\beta+\gamma+2N-n-1\\\alpha+\beta+N-n-1,\alpha+\gamma+N-n\end{array}\right],
\end{array}
\end{equation}
for any positive integer $m$, where, as before, $\alpha=a/\rho$, $\beta=b/\rho$, $\gamma=1/\rho$ and the hypergeometric series $_3F_2$ is evaluated at $1$.

Noticing that
\begin{equation}\label{k1}
K_{(31^n),(1^{n+3})}=\frac{1}{2}\,\frac{(n+2)!}{n!},
\end{equation}
\begin{equation}\label{k2}
K_{(31^n),(21^{n+1})}=n+1,
\end{equation}
it follows from Eq.~(\ref{eqJ1}) that in order to find an expression for $J_{(31^n)}$ we only need to find an explicit formula for $B_{(31^n)}$.

Applying Eq.~(\ref{eqB}) for $\lambda=(41^n)$ and $c=2$ we obtain
\begin{eqnarray}
\rho\left[2(N-n-1)\left(B_{(1^{n+1})}-B_{(31^n)}-B_{(21^{n+1})}\right)-n\left(2B_{(31^n)}+B_{(2^21^{n-1})}\right)\right.  \\
~~~~~ ~~~~~ +\left.nB_{(1^{n+1})}\right]=(b-1)B_{(21^n)}+(b-1)B_{(31^n)}+(a+3)B_{(31^n)}-(a+1)B_{(1^{n+1})}, \nonumber
\end{eqnarray}
namely,
\begin{eqnarray}
B_{(31^n)}=\left[a+b+2+2\rho(N-1)\right]^{-1}\times\left[(a+1+\rho(2N-n-2))B_{(1^{n+1})} \right. \\
-\left.(b-1)B_{(21^n)}-2\rho(N-n-1)B_{(21^{n+1})}-\rho n B_{(2^21^{n-1})}\right]. \nonumber
\end{eqnarray}

Substitution of (\ref{eqb111}) and (\ref{eqb2211}) into this last expression will eventually lead us to the desired formula for $B_{(31^n)}$.

\subsection{Partitions of the form $(32^m)$:} Let $m$ be, again, a non-negative integer number and let us consider the case of the partition $(32^m)$. Despite of the lack of explicit expressions for the corresponding Kostka numbers, it is straightforward to find a recurrence for $B_{(32^m)}$.

In fact, applying Theorem~\ref{thm} for $\lambda=(42^m)$ with $c=3$ we have that
\begin{eqnarray}\label{eq32a1}
B_{(32^m)}=\left[a+b+2+2\rho(N-1)\right]^{-1} \left[(a+2+\rho(2N-m-2))B_{(2^{m+1})} \right. \\
+ \left. \rho(N-m-1)B_{(2^m1^2)}-2\rho(N-m-1)B_{(2^{m+1}1)}\right], \nonumber
\end{eqnarray}
from where a formula for $B_{(32^m)}$ can be read.

\subsection{Partitions of the form $(32^m1^n)$:} Finally let us discuss the case of partitions of the form $(32^m1^n)$ with $m,n > 0$. Applying again Theorem~\ref{thm} but now for $\lambda=(42^m1^n)$ and $c=3$ we get:
\begin{eqnarray}\label{eq32a1bis}
B_{(32^m1^n)}=\left[a+b+2+2\rho(N-1)\right]^{-1} \left[(a+2+\rho(2N-m-n-2))B_{(2^{m+1}1^n)} \right. \\
+ \left. \rho(N-m-n-1)B_{(2^m1^{n+2})}-n\rho B_{(2^{m+2}1^{n-1})}-2\rho(N-m-n-1)B_{(2^{m+1}1^{n+1})}\right]. \nonumber
\end{eqnarray}
As before, using (\ref{eqb111}) and (\ref{eqb2211}) an explicit formula for $B_{(32^m1^n)}$ can be derived.

\bigskip

{\bf Acknowledgments:} Both authors are grateful to the referees for their stimulating reports. SI would like to thank the High Energy Group of the Abdus Salam ICTP for the warm hospitality during the completion of this work. 


\end{document}